\documentclass[a4paper,11pt]{article}
\usepackage{bm,mathrsfs,amsmath,amsthm,amssymb,eepic,amscd,longtable,array,graphicx}
\newcommand{\nc}{\newcommand}
\nc{\rnc}{\renewcommand}
\nc{\nn}{\nonumber}
\nc{\der}{{\partial}}
\rnc{\Im}{{\rm{Im}\,}}
\rnc{\Re}{{\rm{Re}\,}}
\nc{\db}{\displaybreak[0]\\}
\nc{\bra}{\langle}
\nc{\ket}{\rangle}
\nc{\bs}{\boldsymbol}

\DeclareMathOperator{\Tr}{Tr}

\DeclareMathOperator{\End}{End}

\newtheorem{theorem}{Theorem}[section]
\newtheorem{lemma}[theorem]{Lemma}

\newtheorem{proposition}[theorem]{Proposition}

\theoremstyle{definition}
\newtheorem{definition}[theorem]{Definition}

\numberwithin{equation}{section}

\numberwithin{equation}{section}

\begin{document}%
%
\title{Partition functions of integrable lattice models and
combinatorics of symmetric polynomials}

\author{
Kohei Motegi$^1$\thanks{E-mail: motegi@gokutan.c.u-tokyo.ac.jp} \,
Kazumitsu Sakai$^2$\thanks{E-mail: sakai@gokutan.c.u-tokyo.ac.jp} \,
and
Satoshi Watanabe$^2$\thanks{E-mail: watanabe@gokutan.c.u-tokyo.ac.jp}
\\\\
$^1${\it Faculty of Marine Technology, Tokyo University of Marine Science and Technology,}\\
 {\it Etchujima 2-1-6, Koto-Ku, Tokyo, 135-8533, Japan} \\
\\
$^2${\it Institute of physics, University of Tokyo,} \\ 
{\it Komaba 3-8-1, Meguro-ku, Tokyo 153-8902, Japan}
\\\\
\\
}

\date{\today}

\maketitle

\begin{abstract}
We review and present new studies
on the relation between the partition functions of integrable lattice models
and symmetric polynomials, and its combinatorial representation theory
based on the correspondence, including our work.
In particular, we examine the correspondence between
the wavefunctions of the XXZ type, Felderhof type and the boson type
integrable models and symmetric polynomials such as the Schur,
Grothendieck and symplectic Schur functions.
We also give a brief report of our work on generalizing the
correspondence between the Felderhof models and factorial Schur and
symplectic Schur functions.
\end{abstract}

\section{Introduction}
Integrable lattice models
\cite{Bethe,FST,Baxter,KBI}
are special classes of statistical mechanics,
which interesting connections with many subjects of mathematics
have been found in the past, and will be found in the future.
In particular, it plays a very important role in representation theory
and combinatorics, and many notions and objects
in modern representation theory have their origin in integrable models,
the quantum group \cite{Dr,J} for example.
Many objects were introduced by investigating the mathematical structures
of microscopic or quasimacroscopic
quantities such as a single $R$-matrices and monodromy matrices.
From the point of view of statistical mechanics,
the most fundamental quantity in statistical mechanics
is the partition function, which is the most macroscopic bulk quantity
constructed from the monodromy matrices
and characterizes the whole system. For non-integrable models,
partition functions are just numbers, which do not usually have
interesting connections with mathematics.
The situation changes for integrable models:
one can introduce the spectral parameter.
There are many advantages of introducing spectral parameter.
One thing from the point of physics
is that if one introduces the spectral parameter,
we can construct a generating function
of conserved quantities for integrable models.
One advantage from the viewpoint of mathematics
is that we find it corresponds to the symmetric variables
of symmetric polynomials when one investigates a class
of partition functions called the wavefunctions.
Another advantage is that it plays a role of the variable
for the generating function of the enumeration of alternating sign matrices
\cite{Bre,Ku,Ku2,Ok}.
In fact, these facts are deeply related with each other.

In this article, we review and present new studies
on the relation between partition functions of integrable lattice models
and combinatorics of symmetric polynomials.
We mainly deal with two integrable vertex models:
(i) the XXZ-type models and (ii) the Felderhof models.
These models are related with the quantum group of Drinfeld-Jimbo type
and the colored representation, respectively.
In the next two sections,
We explain how symmetric polynomials
such as the Schur, Grothendieck polynomials and their generalizations
arise from a particular type of partition functions
called the wavefunctions.
In section 4, we review how Cauchy and dual Cauchy identities
can be derived by dealing with scalar products and
domain wall boundary partition functions by taking
the XXZ-type and the Felderhof models as an example, respectively.
In section 5, we make comments on other six-vertex models
and give a brief report of our study on
investigating symmetric polynomials from the
generalized Felderhof models.
Section 6 is devoted to conclusion.

\section{XXZ-type models}
We first investigate the Yang-Baxter integrability associated with
the $U_q(sl_2)$ quantum group, following the Appendix of \cite{MS}.

The most fundamental object in quantum integrable models
is the $R$-matrix satisfying the Yang-Baxter relation
\begin{align}
R_{ab}(u_1/u_2)R_{aj}(u_1)R_{bj}(u_2)
=R_{bj}(u_2)R_{aj}(u_1)R_{ab}(u_1/u_2), \label{YBE}
\end{align}
holding in $\mathrm{End}(W_a \otimes W_b \otimes V_j)$
for arbitrary $u_1, u_2 \in \mathbb{C}$.
In this section, we take $W$ and $V$ as
a complex two-dimensional vector spaces $W=V=\mathbb{C}^2$
spanned by the 
``empty state"  $|0\ket={1 \choose 0}$ and the 
``particle occupied state" $|1\ket={0 \choose 1}$,
and take the $R$-matrix
acting on the tensor product $W \otimes V$ or $V \otimes V$
as the following one
\begin{eqnarray}
R(u)=\left( 
\begin{array}{cccc}
u-q u^{-1} & 0 & 0 & 0 \\
0 & q(u-u^{-1}) & 1-q & 0 \\
0 & 1-q & u-u^{-1} & 0 \\
0 & 0 & 0 & u-qu^{-1}
\end{array}
\right), \label{XXZRmatrix}
\end{eqnarray}
which is the $R$-matrix associated with the quantum group $U_q(sl_2)$
\cite{Dr,J}.
The quantum integrable model constructed from the $R$-matrix \eqref{XXZRmatrix}
is called the XXZ chain.

In the original Yang-Baxter relation \eqref{YBE},
every $R$-matrix is the same.
However, one can generalize this relation
to the following $RLL$ relation still keeping the integrability
\begin{align}
R_{ab}(u_1/u_2)L_{aj}(u_1)L_{bj}(u_2)
=L_{bj}(u_2)L_{aj}(u_1)R_{ab}(u_1/u_2), \label{RLL}
\end{align}
holding in $\mathrm{End}(W_a \otimes W_b \otimes V_j)$.
The $L$-operator acts on $W \otimes V$,
where $W$ and $V$ are called the auxiliary and quantum spaces, respectively.
The original Yang-Baxter relation \eqref{YBE}
is recovered by taking the $L$-operator to be
the $R$-matrix $L(u)=R(u)$.
However it is interesting to investigate
the $L$-operator in more detail
since we find that the $L$-operator
obtained in this way generalizes the $R$-matrix \eqref{XXZRmatrix},
which means that one can introduce additional parameters
besides the spectral parameter $u$ and the quantum group $q$.
For example, the parameter $\beta$ which appears in the $L$-operator
below has correspondence in the generalized cohomology theory.

Solving the $RLL$ relation \eqref{RLL} with the $R$-matrix
given by \eqref{XXZRmatrix},
one can show that the following $L$-operator solves the
$RLL$ relation
\begin{align}
L(u)
=
\begin{pmatrix}
\alpha_3 u+\alpha_4 u^{-1} & 0 & 0 & 0 \\
0 & \alpha_3 q u+\alpha_4 u^{-1} & (1-q)\alpha_1 & 0 \\
0 & (1-q)\alpha_2 & \alpha_5 u+\alpha_6 u^{-1} & 0 \\
0 & 0 & 0 & \alpha_5 u+\alpha_6 q u^{-1}
\end{pmatrix}, \label{generalizedLoperator}
\end{align}
where the parameters $\alpha_j, \ j=1,\cdots,6$ and $q$
satisfy the relations
\begin{align}
&(1-q)\alpha_1 \alpha_2+\alpha_3 \alpha_6-\alpha_4 \alpha_5=0, \\
&(q^2-q)\alpha_1 \alpha_2+q^2 \alpha_3 \alpha_6-\alpha_4 \alpha_5=0.
\end{align}
Among the above generalized $L$-operator \eqref{generalizedLoperator},
the following one
\begin{eqnarray}
L(u)=\left( 
\begin{array}{cccc}
u+q \beta u^{-1} & 0 & 0 & 0 \\
0 & q(u+\beta u^{-1}) & 1-q & 0 \\
0 & 1-q & -\beta^{-1}u-u^{-1} & 0 \\
0 & 0 & 0 & -\beta^{-1}u-qu^{-1}
\end{array}
\right), \label{qbetaLoperator}
\end{eqnarray}
which is obtained by setting $\alpha_1=\alpha_2=\alpha_3=1, \alpha_4=q \beta,
\alpha_5=-\beta^{-1}, \alpha_6=-1$
can be regarded as a particularly important $L$-operator.
For example, \eqref{qbetaLoperator}
is a $\beta$-deformation of the $U_q(sl_2)$ $R$-matrix \eqref{XXZRmatrix}.
Another observation is that the wavefunction constructed from the $L$-operator
at $q=0$ gives the $\beta$-Grothendieck polynomials of the
Grassmannian variety.

Now we introduce and examine a class of partition function
which is usually called as the wavefunction.

First, let us consider the monodromy matrix:
\begin{align}
T_{a}(u)=\prod_{j=1}^M L_{a j}(u)
&=
\begin{pmatrix}
A(u) & B(u)  \\
C(u) & D(u)
\end{pmatrix}_{a}. 
\label{monodromy}
\end{align}
The four elements of the monodromy matrix $A(u)$, $B(u)$, $C(u)$ and $D(u)$ are
the operators acting on the quantum space $V_1\otimes \dots \otimes V_M$.

The arbitrary $N$-particle state $|\psi(\{u \}_N) \ket$
(resp. its dual $\langle \psi(\{u \}_N)|$) 
(not normalized) with $N$ spectral parameters
$\{ u \}_N=\{ u_1,u_2,\dots,u_N \}$
is constructed by a multiple action
of $B$ (resp. $C$) operator on the vacuum state 
$|\Omega \rangle:=| 0^{M} \rangle:=|0\ket_1\
\otimes \dots \otimes |0\ket_M$
(resp. $\langle \Omega|:=\langle 0^{M}|:=
{}_1\bra 0|\otimes\dots \otimes{}_M\bra 0|$):
\begin{align}
|\psi(\{u \}_N) \rangle=\prod_{j=1}^N B(u_j)| \Omega \rangle,
\quad
\langle \psi(\{u \}_N)|=\langle \Omega| \prod_{j=1}^N
C(u_j).
\label{statevector}
\end{align}

Next, we introduce the wavefunction
$\bra x_1 \cdots x_N | \psi(\{u\}_N) \ket$ and its dual 
$\bra \psi(\{u\}_N)|x_1\dots x_N \ket$
as the overlap between an arbitrary off-shell 
$N$-particle state $|\psi(\{u\}_N)\ket$ and 
the (normalized) state with an arbitrary particle configuration 
$|x_1 \cdots x_N\ket$ $(1 \le x_1<\dots<x_N \le M$), 
where $x_j$ denotes the positions of the particles. 
The particle configurations are explicitly defined as
\begin{align}
   \bra x_1 \cdots x_N|=\bra \Omega|\prod_{j=1}^N \sigma^+_{x_j}, \qquad
|x_1 \cdots x_N\ket=\prod_{j=1}^N \sigma^-_{x_j} |\Omega\ket.
\end{align}
We find the following form on the wavefunction.
\begin{theorem}
The wavefunction $\bra x_1 \cdots x_N | \psi(\{u\}_N) \ket$
constructed from the generalized $L$-operator \eqref{qbetaLoperator}
has the following form
\begin{align}
\langle x_1 \cdots x_N|\psi(\{ u \}_N) \rangle
=&
\prod_{j=1}^N \frac{(1-q)(u_j+q \beta u_j^{-1})^M}{(-\beta^{-1} u_j-u_j^{-1})}
\prod_{1 \le j<k \le N} \frac{q-u_j^{-2} u_k^2}{1-u_j^{-2} u_k^2} \nonumber \\
&\times \sum_{\sigma \in S_N}\prod_{\substack{1 \le j<k \le N \\ \sigma(j)>\sigma(k)}}
\frac{1-q u_{\sigma(j)}^2 u_{\sigma(k)}^{-2}}
{q-u_{\sigma(j)}^2 u_{\sigma(k)}^{-2}}
\prod_{j=1}^N
\Bigg(
\frac{-\beta^{-1}u_{\sigma(j)}-u_{\sigma(j)}^{-1}}{u_{\sigma(j)}
+q \beta u_{\sigma(j)}^{-1}}
\Bigg)^{x_j}. \label{qbetaGrothendieckpolynomials}
\end{align}
\end{theorem}

We can regard this wavefunction
as a $q$-deformed $\beta$-Grothendieck polynomials since
this generalizes the expression for the wavefunction
constructed from the $L$-operator at the point $q=0$
\cite{MS2}, which gives the $\beta$-Grothendieck polynomials.
The sum of the expression
is exactly proved from the method of matrix product representation
\cite{GMmat,KM} and the domain wall boundary partition function \cite{Ko,Iz}.
A similar proof is given for the wavefunction of the
Felderhof model in the next section. Details for the XXZ-type models
will appear elsewhere \cite{Mo}.

We check below that the wavefunction \eqref{qbetaGrothendieckpolynomials}
reduces essentially to the $\beta$-Grothendieck polynomials.
\begin{align}
\langle x_1 \cdots x_N|\psi(\{ u \}_N) \rangle
=&
\prod_{j=1}^N \frac{u_j^M}{-\beta^{-1} u_j-u_j^{-1}}
\prod_{1 \le j<k \le N} \frac{-u_j^{-2} u_k^2}{1-u_j^{-2} u_k^2} \nonumber \\
&\times \sum_{\sigma \in S_N}\prod_{\substack{1 \le j<k \le N \\ \sigma(j)>\sigma(k)}}
\frac{1}{-u_{\sigma(j)}^2 u_{\sigma(k)}^{-2}}
\prod_{j=1}^N
\Bigg(
\frac{-\beta^{-1}u_{\sigma(j)}-u_{\sigma(j)}^{-1}}{u_{\sigma(j)}}
\Bigg)^{x_j} \nonumber \\
=&
\prod_{j=1}^N \frac{u_j^{M+1}}{-\beta^{-1} u_j^2-1}
\prod_{1 \le j<k \le N} \frac{u_k^2}{u_k^2-u_j^2} \nonumber \\
&\times \sum_{\sigma \in S_N}
\mathrm{sgn}(\sigma)
\prod_{j=1}^N
\prod_{\substack{1 \le j<k \le N \\ \sigma(j)>\sigma(k)}}
u_{\sigma(j)}^{-2} u_{\sigma(k)}^{2}
\prod_{j=1}^N
(-\beta^{-1}-u_{\sigma(j)}^{-2})^{x_j}
\nonumber
\end{align}
\begin{align}
=&
\prod_{j=1}^N \frac{u_j^{M-1}}{-\beta^{-1} u_j^2-1}
\prod_{1 \le j<k \le N} \frac{1}{u_k^2-u_j^2} \nonumber \\
&\times \sum_{\sigma \in S_N}
\mathrm{sgn}(\sigma)
\prod_{j=1}^N (u_j^2)^j
\prod_{\substack{1 \le j<k \le N \\ \sigma(j)>\sigma(k)}}
u_{\sigma(j)}^{-2} u_{\sigma(k)}^{2}
\prod_{j=1}^N
(-\beta^{-1}-u_{\sigma(j)}^{-2})^{x_j} \nonumber \\
=&\frac{\prod_{j=1}^N u_j^{M-1} (-\beta^{-1} u_j^2-1)^{-1}}
{\prod_{1 \le j < k \le N}(u_k^2-u_j^2)}
\mathrm{det}_N(u_j^{2k}(-\beta^{-1}-u_j^{-2})^{x_k})
\nonumber \\
=&(-\beta)^{-N(N-1)/2}\prod_{j=1}^N u_j^{M-1} 
G_\lambda(\bs{z};\beta). \label{correspondence}
\end{align}
Here $G_{\lambda}(\bs{z};\beta)$ is the
Grothendieck polynomials of the Grassmannian variety
\cite{LS,FK,Buch,IN,IS,Mc}
\begin{align}
G_\lambda(\bs{z};\beta)=
   \frac{\mathrm{det}_N(z_j^{\lambda_k+N-k}(1+\beta z_j)^{k-1})}
        {\prod_{1 \le j < k \le N}(z_j-z_k)},
 \label{GR}
\end{align}
where $\bs{z}=\{z_1,\dots,z_N\}$ is a set of variables and
 $\lambda$ denotes a Young diagram
$\lambda=(\lambda_1,\lambda_2,\dots,\lambda_N)$ with weakly decreasing
nonnegative integers $\lambda_1 \ge \lambda_2 \ge \dots \ge \lambda_N \ge 0$.
The correspondence between the Young diagram and the configuration of particles
is given by $\lambda_j=x_{N-j+1}-N+j-1$.
We also relate the symmetric variables $\bs{z}$ and the
spectral parameters by $z_j=-\beta^{-1}-u_j^{-2}$.

We remark that the overall factor of the right hand side
of the correspondence between the wavefunction and the
$\beta$-Grothendieck polynomials \eqref{correspondence}
can easily corrected to be one by a simple gauge transformation
of the $L$-operator, which has combinatorial description
in terms of pipedream \cite{FK},
excited Young diagrams \cite{IN},
set-valued tableaux \cite{Mc} and so on
in the world of Schubert calculus.

\section{Felderhof models}
We review the studies on the Felderhof model
\cite{Mu,DA,FCWZ,BBF,BMN}
.
We start from the following $L$-operator
\begin{eqnarray}
L(u,p,q)=\left( 
\begin{array}{cccc}
1-pqu & 0 & 0 & 0 \\
0 & -p^2(1-p^{-1}qu) & 1-q^2 & 0 \\
0 & (1-p^2)u & u-p^{-1}q & 0 \\
0 & 0 & 0 & u-pq
\end{array}
\right), \label{generalizedrmatrix}
\end{eqnarray}
The generalized $R$-matrix \eqref{generalizedrmatrix}
can be shown to satisfy the Yang-Baxter relation
\begin{align}
R_{ab}(z_1/z_2,p_1,p_1)R_{aj}(z_1,p_1,p_2)R_{bj}(z_2,p_1,p_2)
=R_{bj}(z_2,p_1,p_2)R_{aj}(z_1,p_1,p_2)R_{ab}(z_1/z_2,p_1,p_1),
\label{generalizedyangbaxter}
\end{align}
holding in $\mathrm{End}(W_a \otimes W_b \otimes V_j)$.

The generalized $R$-matrix
\eqref{generalizedrmatrix}
can be constructed from
a class of an exotic quantum group called the
colored representation or the nilpotent representation \cite{Mu,DA}.
The colored representation becomes a finite-dimensional
highest weight representation when the
parameter of the quantum group is fixed at roots of unity.
Each colored representation space is allowed to have
a free parameter, and since the $R$-matrix is understood as
an intertwiner acting on the tensor product of two representation spaces,
one can include two free parameters $p$ and $q$, which can be regarded
to be associated with the auxiliary and quantum spaces respectively.

Let us state the explicit form of the wavefunction
constructed from the generalized $R$-matrix \eqref{generalizedrmatrix}.
\begin{theorem} \label{felderhofwavefunction}
The wavefunction of the generalized $R$-matrix \eqref{generalizedrmatrix}
has the following form
\begin{align}
&\langle x_1 \cdots x_N|\psi(\{ u \}_N) \rangle \nonumber \\
&=
\prod_{j=1}^N (1-q^2)(1-pqu_j)^{M-1} \prod_{1 \le j < k \le N}
\frac{u_k-p^2 u_j}{u_k-u_j} \mathrm{det}_N
\Bigg( \Bigg( \frac{u_j-p^{-1}q}{1-pqu_j} \Bigg)^{x_k-1} \Bigg),
\end{align}
or, in terms of Young diagrams
\begin{align}
&\langle x_1 \cdots x_N|\psi(\{ u \}_N) \rangle \nonumber \\
&=
\prod_{j=1}^N (1-q^2)(1-pqu_j)^{M-1} \prod_{1 \le j < k \le N}
\frac{u_k-p^2 u_j}{u_j-u_k} \mathrm{det}_N
\Bigg( \Bigg( \frac{u_j-p^{-1}q}{1-pqu_j} \Bigg)^{\lambda_k+N-k} \Bigg),
\end{align}
where $\lambda_j=x_{N-j+1}-N+j-1$,
which reduces to the ordinary Schur polynomials
for the case $q=0$ \cite{BBF},
and is also a special case of the factorial Schur functions \cite{BMN}
by an appropriate transformation of variables.
\end{theorem}
To explain how to derive a dual Cauchy identity
from the domain wall boundary partition function in the next section,
we also state the following theorem
on the overlap between the wavefunction
$\langle 1 \cdots M|B(u_1) \cdots B(u_N)$
constructed by acting $B$-operators on the state
$\langle 1 \cdots M|:=\langle 1^{M}|:=
{}_1\bra 1|\otimes\dots \otimes{}_M\bra 1|$,
and the the (normalized) state with an arbitrary hole configuration 
$|\overline{x_1} \cdots \overline{x_N} \ket$
$(1 \le \overline{x_1}<\dots< \overline{x_N} \le M$), 
where $\overline{x_j}$ denotes the positions of the holes. 
Explicitly,
\begin{align}
|\overline{x_1} \cdots \overline{x_N} \ket
=\prod_{j=1}^N \sigma^+_{x_j}
(|1 \rangle_1 \otimes \cdots \otimes |1 \rangle_M),
\end{align}
\begin{theorem}
The wavefunction of the generalized $R$-matrix \eqref{generalizedrmatrix}
has the following form
\begin{align}
&\langle 1 \cdots M|B(u_1) \cdots B(u_N)
| \overline{x_1} \cdots \overline{x_N} \rangle \nonumber \\
&=
\prod_{j=1}^N (1-q^2) \{ -p^2(1-p^{-1}qu_j) \}^{M-1} \prod_{1 \le j < k \le N}
\frac{p^2 u_k-u_j}{p^2(u_k-u_j)} \mathrm{det}_N
\Bigg( \Bigg( \frac{u_j-pq}{-p^2(1-p^{-1}qu_j)} \Bigg)^{\overline{x_k}-1} \Bigg),
\end{align}
or, in terms of Young diagrams
\begin{align}
&\langle 1 \cdots M|B(u_1) \cdots B(u_N)
| \overline{x_1} \cdots \overline{x_N} \rangle \nonumber \\
&=
\prod_{j=1}^N (1-q^2) \{ -p^2(1-p^{-1}qu_j) \}^{M-1} \prod_{1 \le j < k \le N}
\frac{p^2 u_k-u_j}{p^2(u_j-u_k)} \mathrm{det}_N
\Bigg( \Bigg( \frac{u_j-pq}{-p^2(1-p^{-1}qu_j)} \Bigg)^{\overline{\lambda_k}+N-k} \Bigg),
\end{align}
where $\overline{\lambda_j}=\overline{x_{N-j+1}}-N+j-1$.
\end{theorem}
The above theorems can be proved by
combining the matrix product method
and the domain wall boundary partition function, as in \cite{MS2}.
Let us prove Theorem \ref{felderhofwavefunction}
The strategy is as follows.
We first rewrite the wavefunctions
into a matrix product form, following \cite{GMmat}.
The matrix product form can be expressed as a determinant with some overall
factor which remains to be calculated. 
The information of the particle configuration 
$\{x_1,x_2,\dots,x_N \}$ is encoded in the determinant.
On the other hand, the overall factor is independent of the
particle positions, and therefore we can determine this factor by
considering the specific configuration: we 
explicitly evaluate the overlap of the consecutive configuration (i.e. $x_j=j$)
which is essentially the same with the domain wall boundary partition function.

Let us begin to compute the wavefunction
$\bra x_1 \cdots x_N | \psi(\{u\}_N) \ket$.
We first rewrite it into the matrix product representation.
With the help of graphical description,
one finds that the wavefunction can be written as
\begin{align}
\bra x_1 \cdots x_N | \prod_{j=1}^N B(u_j)| \Omega \ket
=\Tr_{W^{\otimes N}}
\left[ Q
\bra x_1 \cdots x_N | \prod_{a=1}^N T_{a}(u_a) |
\Omega \ket
\right],
\label{overlap}
\end{align}
where $Q=| 1^N \rangle \langle 0^N |$
is an operator acting on the tensor product of auxiliary spaces
$W_1\otimes  \dots \otimes W_N$.
The trace here is also over the auxiliary spaces.

Changing the viewpoint of the products of the monodromy matrices, we have
\begin{align}
\prod_{a=1}^N T_{a}(u_a)
=\prod_{j=1}^M \mathcal{T}_j(\{u\}_N),
\end{align}
where $\mathcal{T}_j(\{u\}_N):= 
\prod_{a=1}^N L_{a j}(u_{a}) \in \End( W^{\otimes N} \otimes V_j)$
can be regarded as a monodromy matrix consisting of
$L$-operators acting on the same quantum space $V_j$
(but acting on different auxiliary spaces).  The monodromy matrix
is decomposed as
\begin{align}
\mathcal{T}_j(\{u\}_N)
&:=\begin{pmatrix}
\mathcal{A}_N (\{u\}_N) & \mathcal{B}_N(\{u\}_N) \\
\mathcal{C}_N(\{u\}_N) &  \mathcal{D}_N(\{u\}_N)
\end{pmatrix}_j ,
\label{decomp}
\end{align}
where the elements
($\mathcal{A}_N$, etc.) act on 
$W_1\otimes \dots \otimes W_N$.
The wavefunction \eqref{overlap} can then be rewritten by 
$\mathcal{T}_j(\{u\}_N)$ as
\begin{align}
\bra x_1 \cdots x_N | \psi(\{u\}_N) \ket
&=\Tr_{W^{\otimes N}}
\left[ Q
\bra x_1 \cdots x_N |
\prod_{j=1}^M \mathcal{T}_j(\{u\}_N)
|
\Omega \ket
\right] \nn \\
&=\Tr_{W^{\otimes N}}\left[ Q
\mathcal{D}_N^{M-x_N}
\mathcal{C}_N
\mathcal{D}_N^{x_N-x_{N-1}-1}
\dots\mathcal{C}_N\mathcal{D}_N^{x_2-x_1-1}\mathcal{C}_N\mathcal{D}_N^{x_1-1}
\right].
\label{reov}
\end{align}

For these operators, one finds the following recursive 
relations: 
\begin{align}
&\mathcal{D}_{n+1}(\{u\}_{n+1})
=
\begin{pmatrix}
1-pqu_{n+1} & 0  \\
0 & u_{n+1}-p^{-1}q
\end{pmatrix} 
\otimes
\mathcal{D}_n(\{u\}_n)
+
\begin{pmatrix} 
0 & 0  \\
(1-p^2)u_{n+1} & 0
\end{pmatrix}
\otimes
\mathcal{C}_n(\{u\}_n),
\label{reop1} \\
&\mathcal{C}_{n+1}(\{u\}_{n+1})
=
\begin{pmatrix}
0 & 1-q^2  \\
0 & 0
\end{pmatrix}
\otimes
\mathcal{D}_n(\{u\}_n) 
+
\begin{pmatrix}
-p^2(1-p^{-1}qu_{n+1}) & 0  \\
0 & u_{n+1}-pq
\end{pmatrix}
\otimes
\mathcal{C}_n(\{u\}_n),
\label{reop2}  
\end{align}
with the initial condition
\begin{align}
\mathcal{D}_1=
\begin{pmatrix}
1-pqu_1 & 0 \\
0    &  u_1-p^{-1}q
\end{pmatrix}, \quad
\mathcal{C}_1=
\begin{pmatrix}
0 & 1-q^2 \\
0 & 0 
\end{pmatrix}.
\label{initialCD}
\end{align}

By using the recursive relations \eqref{reop1} and \eqref{reop2},
one sees that these operators satisfy the following simple algebra.
\begin{lemma}\label{algebra}
There exists a decomposition of $\mathcal{C}_n$ :
$\mathcal{C}_n=\sum_{j=1}^n \mathcal{C}_n^{(j)}$ such that
the following algebraic relations hold for $\mathcal{D}_n$ and $\mathcal{C}_n^{(j)}$:
\begin{align}
&\mathcal{C}_n^{(j)}\mathcal{D}_n
=\frac{u_j-p^{-1}q}{1-pq u_j}\mathcal{D}_n \mathcal{C}_n^{(j)}, \label{rel2} \\
&(\mathcal{C}_n^{(j)})^2=0, \label{rel3} \\
&\mathcal{C}_n^{(j)}\mathcal{C}_n^{(k)}
=-
\frac{(pu_j-q)(1-pqu_k)}{(pu_k-q)(1-pqu_j)}
\mathcal{C}_n^{(k)}\mathcal{C}_n^{(j)}, \ \ \ (j \neq k)
\label{rel4}.
\end{align}
\end{lemma}

\begin{proof}
This can be shown by induction on $n$.  For $n=1$,  from \eqref{initialCD}
$\mathcal{D}_1$ is diagonal and one directly sees that the relations are valid.
For $n$, we assume that $\mathcal{D}_n$ is diagonalizable and write the 
corresponding diagonal matrix as $\mathscr{D}_n=G_n^{-1}\mathcal{D}_n G_n$. 
Also writing  $\mathscr{C}_n=G_n^{-1} \mathcal{C}_n G_n$ and 
$\mathscr{C}_n=\sum_{j=1}^{n} \mathscr{C}_n^{(j)}$, and  noting that
the algebraic relations above do not depend on the choice of basis, we suppose by the
induction hypothesis that the same relations are satisfied by $\mathscr{D}_n$
and $\mathscr{C}_n^{(j)}$. 

Now we shall show that they also hold for $n+1$. To this end, first we 
construct $G_{n+1}$. Noting from \eqref{reop1} that $\mathcal{D}_{n+1}$ is an 
upper triangular block matrix whose block diagonal elements are written in 
terms of $\mathcal{D}_n$, 
we assume that $G_{n+1}$ is written as
\begin{equation}
G_{n+1}=
\begin{pmatrix}
G_n &  0 \\
G_n H_n  & G_n
\end{pmatrix},
\label{G-matrix}
\end{equation}
where $2n\times 2n$ matrix $H_n$ remains to be determined. 
Using the induction hypothesis for $n$, one obtains
\begin{align}
&G_{n+1}^{-1}\mathcal{D}_{n+1} G_{n+1} \nonumber \\
=&
\begin{pmatrix}
(1-pqu_{n+1}) \mathscr{D}_n & 0 \\
(u_{n+1}-p^{-1}q)\mathscr{D}_n H_n
+(1-p^2)u_{n+1} \mathscr{C}_n
-(1-pqu_{n+1})H_n \mathscr{D}_n
& (u_{n+1}-p^{-1}q)\mathscr{D}_n
\end{pmatrix}.
\end{align}
The above matrix is guaranteed to be diagonal when 
\begin{equation}
(u_{n+1}-p^{-1}q)\mathscr{D}_n H_n
+(1-p^2)u_{n+1} \mathscr{C}_n
-(1-pqu_{n+1})H_n \mathscr{D}_n=0
\end{equation}
Utilizing the above relation and  recalling  $\mathscr{D}_n$
and $\mathscr{C}^{(j)}_n$ satisfy the relation same as that in \eqref{rel2}, 
one finds
\begin{align}
H_n=\mathscr{D}^{-1}_n\sum_{j=1}^n
\frac{(1-p^2)u_{n+1}(1-pqu_j) }
        {(1-q^2)(u_j-u_{n+1})} \mathscr{C}_n^{(j)}.
\label{H-matrix}
\end{align}
One thus obtains the diagonal matrix $\mathscr{D}_{n+1}$:
\begin{align}
\mathscr{D}_{n+1}=
\begin{pmatrix}
(1-pqu_{n+1})\mathscr{D}_n & 0 \\
0 & (u_{n+1}-p^{-1}q)\mathscr{D}_n
\end{pmatrix}.
\label{D-matrix}
\end{align}
The remaining task is to derive  $\mathscr{C}_{n+1}^{(j)}$ and
to prove the relations \eqref{rel2}--\eqref{rel4} hold for $n+1$.
Combining  \eqref{reop2}, \eqref{G-matrix} and \eqref{H-matrix},
and also inserting the relations \eqref{rel3} and \eqref{rel4},
one arrives at $\mathscr{C}_{n+1}=\sum_{j=1}^{n+1}\mathscr{C}_{n+1}^{(j)}$
where
\begin{align}
\mathscr{C}_{n+1}^{(j)}=
\begin{cases} \displaystyle
\frac{1}{u_j-u_{n+1}}
\begin{pmatrix}
(u_{n+1}-p^2 u_j)(1-pqu_{n+1}) \mathscr{C}_n^{(j)} & 0 \\
0 & (u_{n+1}-p^{-1}q)(p^2 u_j-u_{n+1}) \mathscr{C}_n^{(j)}
\end{pmatrix}  \\[6mm]
\text{ for $1\le j \le n$} \\[6mm]
\begin{pmatrix}
0  & (1-q^2)\mathscr{D}_n \\
0 & 0
\end{pmatrix}   \text{ for $j=n+1$}
\end{cases}.
\label{C-matrix}
\end{align}

Finally recalling that $\mathscr{D}_n$ and $\mathscr{C}_n^{(j)}$ 
are supposed to
satisfy the relations \eqref{rel2}--\eqref{rel4} and using the explicit
form of $\mathscr{D}_{n+1}$ \eqref{D-matrix} and $\mathscr{C}_{n+1}^{(j)}$ 
\eqref{C-matrix}, one sees they satisfy the same algebraic relations as those 
in \eqref{rel2}--\eqref{rel4} for $n+1$.
\end{proof}

Due to the algebraic relations \eqref{rel2}, \eqref{rel3}
and \eqref{rel4} in Lemma~\ref{algebra}, 
the matrix product form for the wavefunction \eqref{reov} can be rewritten
as
\begin{align}
\bra 
x_1 \cdots x_N |\psi(\{u\}_N)
\ket
=&\prod_{j=1}^N \Bigg( \frac{1-pqu_j}{u_j-p^{-1}q} \Bigg)^j
\Tr_{W^{\otimes N}}\left[
Q \mathcal{D}_N^{M-N}
\mathcal{C}_N^{(N)}
\dots\mathcal{C}_N^{(1)} \right] \nonumber \\
&\times \sum_{\sigma \in \mathfrak{S}_N} (-1)^\sigma
    \prod_{j=1}^N
  \Bigg( \frac{u_{\sigma(j)}-p^{-1}q}{1-pq u_{\sigma(j)}} \Bigg)^{x_j},
\label{wavefunctiontochu}
\end{align}
where $\mathfrak{S}_N$ is the symmetric group of order $N$.
One easily notes that \eqref{wavefunctiontochu}
can be further rewritten in the following determinant form:
\begin{align}
\bra 
x_1 \cdots x_N |\psi(\{u\}_N)
\ket
=&K \prod_{j=1}^N \Bigg( \frac{1-pqu_j}{u_j-p^{-1}q} \Bigg)^j
\mathrm{det}_N
 \Bigg( \Bigg( \frac{u_{j}-p^{-1}q}{1-pq u_{j}} \Bigg)^{x_k} \Bigg),
\label{predet}
\end{align}
where the prefactor $K$ given below remains to be determined:
\begin{align}
K=\Tr_{W^{\otimes N}}\left[
Q \mathcal{D}_N^{M-N}
\mathcal{C}_N^{(N)}
\dots\mathcal{C}_N^{(1)} \right].
\end{align}

In \eqref{predet},
we notice that the information of the particle configuration
$\{x_1, x_2,\dots,x_N \}$ is encoded in the determinant,
while the overall factor $K$ is independent of the configuration.
This fact allows us to determine the factor $K$ by evaluating
the overlap for a particular particle configuration. In fact,
we find the following explicit expression
for the case $x_j=j$ ($1\le j \le N$):
\begin{proposition}
The wavefunction $\bra 
x_1 \cdots x_N |\psi(\{u\}_N)
\ket$
for the case $x_j=j$ ($1\le j \le N$)
has the following form:
\begin{align}
\bra 
12 \cdots N |\psi(\{u\}_N)
\ket
=(1-q^2)^{N(N+1)/2} \prod_{j=1}^N
(u_j-p^{-1}q)^{M-N}
\prod_{1 \le j<k \le N}(u_k-p^2 u_j). \label{stepoverlap}
\end{align}
\end{proposition}
\begin{proof}
We can easily show by its graphical description that
$\bra 
12 \cdots N |\psi(\{u\}_N)
\ket$
can be factorized as
\begin{align}
\bra 
12 \cdots N |\psi(\{u\}_N)
\ket
=\prod_{j=1}^N
(u_j-p^{-1}q)^{M-N}
Z_N(\{ u \}_N), \label{stepoverlapfactorization}
\end{align}
where $Z_N(\{ u \}_N)$ is the domain wall boundary partition function.
The domain wall boundary partition function
on an $M \times M$ lattice is defined as
\begin{align}
Z_M(\{ u \}_M)=\langle 1 \cdots M|
B(u_1) \cdots B(u_M)| \Omega \rangle,
\end{align}
where $M$ $B$-operators are inserted between the vacuum vector
$| \Omega \rangle$ and the state of particles
$\langle 1 \cdots M|=
{}_1\bra 1|\otimes\dots \otimes{}_M\bra 1|$.

One can show a more general result
for the domain wall boundary partition function
with inhomogeneities
\begin{align}
Z_M(\{ u \}_M|\{ v \}_M, \{ q \}_M)=\langle 1 \cdots M|
B(u_1|\{ v \}_M, \{ q \}_M) \cdots B(u_M|\{ v \}_M, \{ q \}_M)| \Omega \rangle,
\end{align}
where
\begin{align}
B(u|\{ v \}_M, \{ q \}_M)
={}_a \langle 0 |
L_{a N}(u/v_M, q_M) \cdots
L_{a 1}(u/v_1, q_1)| 1 \rangle_a.
\end{align}

\begin{lemma}{\rm cf. \cite{FCWZ}}
\label{inhomogeneousdomainwall}
The domain wall boundary partition function with inhomogeneities
has the following form.
\begin{align}
Z_M(\{ u \}_M|\{ v \}_M, \{ q \}_M)=\prod_{j=1}^M \frac{1-q_j^2}{v_j^{M-1}}\prod_{1 \le j<k \le M}
(v_k-q_j q_k v_j)(u_k-p^2 u_j). \label{inhomogeneousdomain}
\end{align}
\end{lemma}
Lemma \ref{inhomogeneousdomainwall} can be proved by
using the Izergin-Korepin technique, i.e.,
show that both hand sides of \eqref{inhomogeneousdomain}
satisfy the same recursive relation, initial condition
and the degree counting of polynomials.

Taking the homogeneous limit $q_j \to 1$, $v_j \to 1$ ($j=1,\cdots,M$)
of \eqref{inhomogeneousdomain}
and inserting into \eqref{stepoverlapfactorization}
gives \eqref{stepoverlap}.
\end{proof}
At last, Theorem \ref{felderhofwavefunction}
can be proved by
checking that it has the determinant form
\eqref{predet}
and satisfies the particular case
\eqref{stepoverlap}.

\section{Combinatorial identities}
In this section,
we derive combinatorial identities by investigating partition functions
in more detail.
We show that Cauchy identities are derived from scalar products,
while dual Cauchy identities are obtained from
domain wall boundary partition functions,
which we explain by XXZ-type models and Felderhof models respectively.
\subsection{Cauchy identities}
The scalar product \cite{KBI} between the arbitrary off-shell state vectors
is defined as
\begin{align}
\langle \psi(\{ u \}_N)| \psi(\{ v \}_N) \rangle
=\langle \Omega| \prod_{j=1}^N C(u_j)
\prod_{k=1}^N B(v_k)| \Omega \rangle
\label{SP}
\end{align}
with $u_j, v_k\in \mathbb{C}$.
Here we illustrate a way to derive a Cauchy identity
for the $\beta$-Grothendieck polynomials
from the scalar products of the
$q=0$ limit of the $L$-operator \eqref{betaLoperator}.
\begin{eqnarray}
L(u)=\left( 
\begin{array}{cccc}
u & 0 & 0 & 0 \\
0 & 0 & 1 & 0 \\
0 & 1 & -\beta^{-1}u-u^{-1} & 0 \\
0 & 0 & 0 & -\beta^{-1}u
\end{array}
\right). \label{betaLoperator}
\end{eqnarray}
First, let us recall the following correspondence between the
wavefunction constructed from the \eqref{betaLoperator}.
\begin{theorem} \label{th-wave} {\rm \cite{MS}}
The (off-shell) wavefunction and its dual wave-function
of the integrable five-vertex model \eqref{betaLoperator}
are, respectively, given by the Grothendieck polynomials as
\begin{align}
\bra x_1 \cdots x_N|\psi(\{ u \}_N) \ket&=(-\beta^{-1})^{N(N-1)/2}
\prod_{j=1}^N u_j^{M-1} G_\lambda({\bs z};\beta), 
\label{wavefunctionone} \\
\bra \psi(\{ u \}_N)|x_1 \cdots x_N \ket&=(-\beta^{-1})^{N(N-1)/2}
\prod_{j=1}^N u_j^{M-1} G_{\lambda^\vee}({\bs z};\beta),
\label{wavefunctiontwo}
\end{align}
where $z_j=-\beta^{-1}-u_j^{-2}$,  and
$\lambda=(\lambda_1,\dots,\lambda_N)$
($M-N \ge \lambda_1 \ge \cdots \ge \lambda_N \ge 0$)
and $\lambda^\vee=(\lambda_1^\vee,\dots,\lambda_N^\vee)$
($M-N \ge \lambda_1^\vee \ge \cdots \ge \lambda_N^\vee \ge 0$)
are the Young diagrams related to the particle configuration
$x=(x_1, \dots, x_N) $ as
$\lambda_j=x_{N-j+1}-N+j-1$ and $\lambda_j^\vee=M-N+j-x_j$,
respectively.
\end{theorem}
\noindent
Note that the Young diagram $\lambda^\vee$ is the complementary
part of the Young diagram $\lambda$ in the $N \times (M-N)$ rectangular
Young diagram.

Next, we recall that one
can show the following determinant form \cite{BoS,KBI,MS}.
\begin{theorem}{\label{scalarthm}}
The scalar product is given by a determinant form:
\begin{align}
\bra \psi(\{ u \}_N)| \psi(\{ v \}_N) \ket
=
\prod_{1 \le j <k \le N} \frac{1}{(u_j^2-u_k^2)(v_k^2-v_j^2)}
\mathrm{det}_N Q(\{ u \}_N|\{ v \}_N),
\label{generalscalar}
\end{align}
where $\{u\}_N$ and $\{v\}_N$ are arbitrary sets
of complex values (i.e. off-shell conditions), and
$Q$ 
is an $N \times N$ matrix with matrix elements
\begin{align}
Q(\{ u \}_N|\{ v \}_N)_{jk}=
\frac{ u_j^M (-\beta^{-1}v_k-v_k^{-1})^M v_k^{2(N-1)}
           -v_k^M (-\beta u_j-u_j^{-1})^M u_j^{2(N-1)}
         }
        {v_k/u_j-u_j/v_k}.
\label{element}
\end{align}
\end{theorem}
The Cauchy formula for the $\beta$-Grothendieck polynomials can be
derived by combining Theorem \ref{th-wave} and \ref{scalarthm}.
The key is to substitute the completeness relation,
\begin{align}
\sum_{\{ x \}}|x_1 \cdots x_N \rangle \langle x_1 \cdots x_N |=\mathrm{Id},
\end{align}
and decompose the scalar product as
\begin{align}
\bra \psi(\{u\}_N)|\psi(\{v\}_N)\ket=\sum_{1\le x_1<\cdots<x_N\le M}
\bra \psi(\{u\}_N)|x_1\cdots x_N\ket\bra x_1\cdots x_N |\psi(\{v\}_N)\ket.
\label{SPdecomposition}
\end{align}
Using Theorem \ref{scalarthm} and \ref{th-wave} in the left and right hand side
of the equality \eqref{SPdecomposition} respectively,
one has the following Cauchy identity.
\begin{theorem}{\rm \cite{MS}}
The following Cauchy identity
for the Grothendieck polynomials holds.
\begin{align}
&\sum_{\lambda \subseteq L^N} G_\lambda({\bs z};\beta)
G_{\lambda^\vee}({\bs w};\beta)
\nonumber \\
=&\prod_{1 \le j<k \le N} \frac{1}{(z_j-z_k)(w_k-w_j)}
\mathrm{det}_N
\left[
\frac{z_j^{L+N}(1+\beta w_k)^{N-1}-w_k^{L+N}(1+\beta z_j)^{N-1}}
{z_j-w_k}
\right], \label{cauchy}
\end{align}
where the Young diagram
$\lambda^\vee=(\lambda_1^\vee,\dots,\lambda_N^\vee)$
is given by the Young diagram $\lambda=(\lambda_1,\dots,\lambda_N)$
as $\lambda_j^\vee=L-\lambda_{N+1-j}$.
\end{theorem}
\noindent
Here we have set $L=M-N$, but the above formula holds for any
$L \ge 0$.

\subsection{Dual Cauchy identities}
The dual Cauchy identities can be derived by
dealing with domain wall boundary partition functions \cite{Ko,Iz}
\begin{align}
&\langle 1 \cdots M|
B(u_1) \cdots B(u_M)| \Omega \rangle,
\end{align}
where $M$ $B$-operators are inserted between the vacuum vector
$| \Omega \rangle$ and the state of particles
$\langle 1 \cdots M|=
{}_1\bra 1|\otimes\dots \otimes{}_M\bra 1|$.
This class of partition function has found applications
to the enumeration of alternating sign matrices in the 1990s,
and it was only noticed in recent years to have applications
to the dual Cauchy identities \cite{BBF,BMN}.
We illustrate this by the Felderhof models.

First, we rewrite the wavefunctions in the following forms
\begin{align}
&\langle x_1 \cdots x_N|
B(u_1) \cdots B(u_N)| \Omega \rangle=
\prod_{j=1}^N \frac{(1-q^2)^M}{(1+pqw_j)^{M-1}} \nonumber \\
\times&
\prod_{1 \le j < k \le N}
\frac{
p^2q(p^2-1)w_j w_k+p(p^2-q^2)w_k+p(p^2q^2-1)w_j+q(p^2-1)
}
{
p(q^2-1)
}
s_\lambda({\bs w}), \label{wavefunctionfordualcauchy}
\end{align}
where we make transformation of variables from $u_j$ to
$\displaystyle w_j=\frac{u_j-p^{-1}q}{1-pqu_j}
$, and
\begin{align}
&\langle 1 \cdots M|
B(u_1) \cdots B(u_N)| \overline{x_1} \cdots \overline{x_N} \rangle=
\prod_{j=1}^N \frac{(-p^2)^{M-1}(1-q^2)^M}{(1+p^{-1}qz_j)^{M-1}} \nonumber \\
\times&
\prod_{1 \le j < k \le N}
\frac{
q(1-p^2)z_j z_k+p(q^2-p^2)z_k+p(1-p^2 q^2)z_j+p^2 q(1-p^2)
}
{
-p^5(q^2-1)
}
s_{\overline{\lambda}} \Bigg( \frac{{\bs z}}{-p^2} \Bigg),
\label{wavefunctionfordualcauchytwo}
\end{align}
where we also transform from $u_j$ to
$\displaystyle z_j=\frac{u_j-pq}{1-p^{-1}qu_j}
$.

The dual Cauchy identities is derived by
evaluating the domain wall boundary partition function
in two ways.

First, we evaluate the domain wall boundary partition function
by viewing it as a particular limit of the wavefunction \eqref{wavefunctionfordualcauchy}. One has
\begin{align}
&\langle 1 \cdots M|
B(u_1) \cdots B(u_M)| \Omega \rangle=
\prod_{j=1}^M \frac{(1-q^2)^M}{(1+pqw_j)^{M-1}} \nonumber \\
\times&
\prod_{1 \le j < k \le M}
\frac{
p^2q(p^2-1)w_j w_k+p(p^2-q^2)w_k+p(p^2q^2-1)w_j+q(p^2-1)
}
{
p(q^2-1)
}. \label{comparisonone}
\end{align}

Another way is to insert the completeness relation
\begin{align}
\sum_{\{ x \}}|x_1 \cdots x_N \rangle \langle x_1 \cdots x_N |=\mathrm{Id},
\end{align}
between the $B$-operators
\begin{align}
&\langle 1 \cdots M|
B(u_1) \cdots B(u_M)| \Omega \rangle \nonumber \\
=&\sum_{\{x \}} \langle 1 \cdots M|B(u_1) \cdots
B(u_{M-N}) |x_1 \cdots x_N \rangle \langle x_1 \cdots x_N|
B(u_{M-N+1}) \cdots
B(u_M)|\Omega \rangle \nonumber \\
=&\sum_{\{x \}} \langle 1 \cdots M|B(u_1) \cdots
B(u_{M-N}) |\overline{x_1} \cdots \overline{x_{M-N}} \rangle \langle x_1 \cdots x_N|
B(u_{M-N+1}) \cdots
B(u_M)|\Omega \rangle,
\label{comparisontwo}
\end{align}
and insert
\eqref{wavefunctionfordualcauchy} and \eqref{wavefunctionfordualcauchytwo}
into the right hand side of \eqref{comparisontwo}.
Comparing \eqref{comparisonone} and \eqref{comparisontwo},
we find
\begin{align}
\sum_{\lambda \subseteq (M-N)^N}
s_{\overline{\lambda}}\Bigg(\frac{{\bs z}}{-p^2}\Bigg)
s_\lambda( \overline{{\bs w}})
=
(-p^2)^{(N-M)N} \prod_{j=1}^{M-N} \Bigg( \frac{1+p^{-1}qz_j}{1+pqw_j} \Bigg)^{M-1} \frac{A}{BC}
, \label{dualcauchy}
\end{align}
where
\begin{align}
A=&\prod_{1 \le j < k \le M}
\frac{
p^2q(p^2-1)w_j w_k+p(p^2-q^2)w_k+p(p^2q^2-1)w_j+q(p^2-1)
}
{
p(q^2-1)
}, \\
B=&\prod_{M-N+1 \le j < k \le M}
\frac{
p^2q(p^2-1)w_j w_k+p(p^2-q^2)w_k+p(p^2q^2-1)w_j+q(p^2-1)
}
{
p(q^2-1)
}, \\
C=&\prod_{1 \le j < k \le M-N}
\frac{
q(1-p^2)z_j z_k+p(q^2-p^2)z_k+p(1-p^2 q^2)z_j+p^2 q(1-p^2)
}
{
-p(q^2-1)
},
\end{align}
and $\overline{{\bs w}}=\{w_{M-N+1},\dots,w_M \}$,
${\bs z}=\{z_1,\dots,z_{M-N} \}$.
Note also that the sum over all particle configurations $\{x \}$
is translated to the sum over all Young diagrams $\lambda$
satisfying $\lambda \subseteq (M-N)^N$.

\eqref{dualcauchy} is nothing but the dual Cauchy formula
for the Schur functions. In fact, if we set $q=0$ and $t=-p^2$,
\eqref{dualcauchy} becomes
\begin{align}
\sum_{\lambda \subseteq (M-N)^N}
s_{\overline{\lambda}}\Bigg(\frac{{\bs u}}{t}\Bigg)
s_\lambda( \overline{{\bs u}})
=
\prod_{j=1}^{M-N} \prod_{k=M-N+1}^M \Bigg(\frac{u_j}{t}+u_k \Bigg),
\end{align}
which becomes the celebrated dual Cauchy identity
by scaling ${\bs u}$ to $t {\bs u}$ \cite{BBF}.
See \cite{BMN,BW,BWZ,MSW}
for example for more results on this direction of research
of deriving other combinatorial identities,
by changing boundary conditions
for example.

\section{Other models, formulae and generalizations}
We give several remarks on  other integrable models.
\subsection{Combinatorial formula for the Schur polynomials}
First, we remark that there are other interesting six-vertex models.
For example, the following $L$-operator
\begin{align}
\mathcal{L}_{aj}(v)
=
\begin{pmatrix}
1-\beta v & 0 & 0 & 0 \\
0 & 1+\beta v & 2v & 0 \\
0 & 1 & v & 0 \\
0 & 0 & 0 & v
\end{pmatrix}_{aj}.
\label{sixvertexLoperator}
\end{align}
can be shown to be integrable.
In fact this $L$-operator is another special limit of the
generalized XXZ-type six-vertex models in section 2.
By investigating the wavefunction itself in detail,
one can find the following combinatorial formula
for the Schur polynomials.
\begin{theorem} \label{combinatorialformula} {\rm \cite{MS3}}
We have the following combinatorial formula for the Schur polynomials
\begin{align}
s_\lambda(\bs{z})=&
\frac{1}{\prod_{1 \le j < k \le N}(z_j+z_k+2 \beta z_j z_k)}
\sum_{x^{(N)} \succ x^{(N-1)} \succ \dots \succ x^{(0)}=\phi} \prod_{k=1}^N
\Bigg\{z_k^{\sum_{j=1}^k x_j^{(k)}-\sum_{j=1}^{k-1} x_j^{(k-1)}-1} \nonumber \\
\times&\left( \frac{2(1+\beta z_k)}{1+2 \beta z_k} \right)^{\#(x^{(k)}|x^{(k-1)})-1}
\prod_{j=1}^{k-1}
\left(
1+2 \beta z_k(1-\delta_{x_j^{(k-1)} x_{j+1}^{(k)}})
\right)
\Bigg\},
\end{align}
where $\beta$ is an arbitrary parameter.
$x^{(k)}=(x_1^{(k)},\dots,x_k^{(k)})$, $k=0,1,\dots,N$ are strict partitions
satisfying the interlacing relations
$x^{(N)} \succ x^{(N-1)} \succ \dots \succ x^{(0)}=\phi$,
and
$x^{(N)}$ is fixed by the Young diagram
$\lambda=(\lambda_1,\dots,\lambda_N)$ as
$\lambda_j=x_j^{(N)}-N+j-1$,
and $\#(y|x)$ denotes
the number of parts in $y$ which are not in $x$.
\end{theorem}
This type of formula of expressing Schur polynomials
using an additional parameter resembles, but is different from
the Tokuyama formula \cite{To,OkTo}.
The modern understanding of the Tokuyama formula comes
from the fact that Schur $Q$-functions factorizes
into an overall factor and Schur functions when the length
of the Young diagram which labels the Schur $Q$-functions
is the same with the number of symmetric variables
\cite{HK}. The proof of Theorem \ref{combinatorialformula}
also relies essentially on this fact.

\subsection{Boson model}

In this section, we remark the relation between the wavefunction of
an integrable boson model (nonhermitian phase model) \cite{BN}
and the Grothendieck polynomials.
The nonhermitian phase model is
a boson system characterized by the generators 
$\phi$, $\phi^\dagger$, $N$ and $\pi$ acting on a bosonic 
Fock space $\mathcal{F}$ spanned by orthonormal basis  
$| n \ket \ (n=0,1,\dots, \infty)$. Here the number $n$ 
indicates the occupation number of bosons.  The generators 
$\phi$, $\phi^\dagger$, $N$ and $\pi$ are, respectively, 
the annihilation, creation, number and vacuum projection operators,
whose actions on $\mathcal{F}$ are, respectively, defined as
\begin{align}
\phi|0 \ket=0,  \quad 
\phi|n \ket=|n-1 \ket, \quad
\phi^\dagger|n \ket=|n+1 \ket, \quad 
N|n \ket=n|n \ket, \quad
\pi|n \ket=\delta_{n\,0}|n \ket.
\end{align}
Thus the operator forms are explicitly given by
\begin{align}
\phi=\sum_{n=0}^\infty |n \ket \bra n+1|, \quad
\phi^\dagger=\sum_{n=0}^\infty |n+1 \ket \bra n|, \quad
N=\sum_{n=0}^\infty n |n \ket \bra n|, \quad
\pi=|0 \ket \bra 0|.
\end{align}
These operators generate an algebra referred to as the phase algebra:
\begin{align}
[\phi, \phi^\dagger]=\pi, \quad
[N, \phi]=-\phi, \quad
[N, \phi^\dagger]=\phi^\dagger.
\end{align}
We consider the tensor product of Fock spaces
$\otimes_{j=0}^{M-1} \mathcal{F}_j$, whose basis is given by
$|\{ n \}_M \ket :=
\otimes_{j=0}^{M-1}|n_j \ket_j$, $n_j=0,1,\dots,\infty$.
We denote a dual state of $|\{ n \}_M \ket$ as
$\bra \{ n \}_M| := \otimes_{j=0}^{M-1} {}_j \bra n_j|$.
The operators $\phi_j$, $\phi_j^\dagger$, $N_j$ and $\pi_j$
act on the Fock space $\mathcal{F}_j$ as $\phi$, 
$\phi^\dagger$, $N$ and $\pi$,
and the other Fock spaces $\mathcal{F}_k, k \neq j$ as
an identity.

The $L$-operator for the nonhermitian phase model \cite{BN} is given by
\begin{align}
\mathcal{L}_{a j}(v)=
\begin{pmatrix}
v^{-1}-\beta v \pi_j & \phi_j^\dagger  \\
\phi_j & v
\end{pmatrix},
\label{Lop-boson}
\end{align}
acting on the tensor product $W_a \otimes \mathcal{F}_j$ of the
complex two-dimensional space $W_a$ and the Fock space at the
$j$th site $\mathcal{F}_j$.
The $L$-operator satisfies the intertwining relation ($RLL$-relation)
\begin{align}
R_{ab}(u/v)\mathcal{L}_{a j}(u)\mathcal{L}_{b j}(v)=
\mathcal{L}_{b j}(v)\mathcal{L}_{a j}(u)R_{ab}(u/v),
\label{RLL-boson}
\end{align}
which acts on $W_a \otimes W_b \otimes \mathcal{F}_j$.
The $R$-matrix $R(u)$ is the same as the one for the integrable
five-vertex model which is a $q=0$ limit of the $U_q(sl_2)$ $R$-matrix
\eqref{XXZRmatrix}. The auxiliary space $W_a$ is the
complex two-dimensional space, which is the same as that for the
integrable five-vertex model, while the quantum space $\mathcal{F}_j$
is the infinite-dimensional bosonic Fock space.

From the $L$-operator, we construct the monodromy matrix
\begin{align}
\mathcal{T}_{a}(v)=\mathcal{L}_{a M-1}(v) \cdots \mathcal{L}_{a 0}(v)
=
\begin{pmatrix}
\mathcal{A}(v) & \mathcal{B}(v)  \\
\mathcal{C}(v) & \mathcal{D}(v)
\end{pmatrix}_{a}, 
\label{bosonmonodromy}
\end{align}
which acts on $W_a \otimes (\mathcal{F}_0\otimes\dots\otimes 
\mathcal{F}_{M-1})$.

The arbitrary $N$-particle state $|\Psi(\{v \}_N) \ket$
(resp. its dual $\bra \Psi(\{v \}_N)|$) 
(not normalized) with $N$ spectral parameters
$\{ v \}_N=\{ v_1,\dots,v_N \}$
is constructed by a multiple action
of $\mathcal{B}$ (resp. $\mathcal{C}$) operator on the vacuum state 
$|\Omega \ket:= | 0^{M} \ket:=|0\ket_0\
\otimes \dots \otimes |0\ket_{M-1}$
(resp. $\bra \Omega| := \bra 0^{M}|:=
{}_0\bra 0|\otimes\dots \otimes{}_{M-1}\bra 0|$):
\begin{align}
|\Psi(\{v \}_N) \ket=\prod_{j=1}^N \mathcal{B}(v_j)| \Omega \ket,
\quad
\bra \Psi(\{v \}_N)|=\bra \Omega| \prod_{j=1}^N
\mathcal{C}(v_j).
\label{statevector-B}
\end{align}

The orthonormal basis of
the $N$-particle state $|\Psi(\{v \}_N) \ket$
and its dual $\bra \Psi(\{v \}_N)|$
is given by $| \{ n \}_{M,N} \ket := |n_0\ket_0\
\otimes \dots \otimes |n_{M-1}\ket_{M-1}$
and $\bra \{ n \}_{M,N}| := _{0}\bra n_0| \otimes \dots
\otimes _{M-1} \bra n_{M-1}|$, where $n_0+n_1+\cdots+n_{M-1}=N$.
The wavefunctions can be expanded in this basis as
\begin{align}
&|\Psi(\{v \}_N) \ket
=\sum_{\substack{0 \le n_0,\dots,n_{M-1} \le N \\
n_0+\cdots+n_{M-1}=N}}
\bra \{ n \}_{M,N}|\psi(\{v \}_N) \ket
|\{ n \}_{M,N} \ket, \\
&\bra \Psi(\{v \}_N)|
=\sum_{\substack{0 \le n_0, \dots,n_{M-1} \le N \\
n_0+\cdots+n_{M-1}=N}}
\bra \{ n \}_{M,N}| \bra \psi(\{v \}_N)| \{ n \}_{M,N} \ket.
\end{align}
There is a one-to-one correspondence between
the set $\{ n \}_{M,N}=\{n_0,n_1,\dots,n_{M-1} \}$ ($n_0+n_1+\cdots+n_{M-1}=N$)
and the Young diagram $\lambda=(\lambda_1,\lambda_2,\dots,\lambda_N)$
($M-1 \ge \lambda_1 \ge \lambda_2 \ge \cdots \ge \lambda_N \ge 0$).
Namely, each Young diagram $\lambda$ under the constraint
$\ell(\lambda) \le N$, $\lambda_1 \le M-1$
can be labeled by a set of integers $\{ n \}_{M,N}$ as
$\lambda=((M-1)^{n_{M-1}}, \dots, 1^{n_1},0^{n_0})$.
We have the following correspondence between the wavefunctions
and the $\beta$-Grothendieck polynomials.
\begin{theorem} \label{bosontheorem} {\rm \cite{MS}}
The wavefunctions can be expressed by the Grothendieck polynomials as
\begin{align}
& \bra \{ n \}_{M,N}|\Psi(\{v \}_N) \ket
=\prod_{j=1}^N (v_j^{-1}-\beta v_j)^{M-1}G_\lambda(z_1,\dots,z_N;\beta), 
\label{Grothen-B1} \\
&\bra \Psi(\{v \}_N)| \{ n \}_{M,N} \ket
=\prod_{j=1}^N (v_j^{-1}-\beta v_j)^{M-1}
G_{\lambda^\vee}(z_1,\dots,z_N;\beta) \label{Grothen-B2},
\end{align}
where $z_j^{-1}=v_j^{-2}-\beta$ and $\lambda^\vee=(\lambda_1^\vee,\lambda_2^\vee,\dots,\lambda_N^\vee)$ ($M-1 \ge \lambda_1^\vee \ge \cdots \ge \lambda_N^\vee \ge 0$) is given by the Young diagram $\lambda$ as
$\lambda_j^\vee=M-1-\lambda_{N+1-j}$.
\end{theorem}

We remark that there exists a $q$-deformation of the
nonhermitian phase model introduced in \cite{SW}.
The six-vertex model \eqref{sixvertexLoperator} 
in the former section can also be regarded
as a certain degeneration point $q=i$
of the integrable $q$-boson $L$-operator.
We also remark that besides the quantum inverse scattering approach,
there is another approach to this model
based on the affine Hecke algebra \cite{Ta,Ta2}.
See \cite{MS,BN,SW,Bo,SU,Tsilevich,KS,Ta,Ta2}
for more about the integrable boson models,
the correspondence
between the wavefunctions and the
Grothendieck, Hall-Littlewood polynomials and their generalizations,
and its relation with the Verlinde algebra etc, for example.

\subsection{Generalized Felderhof models
and Generalized factorial Schur functions}

Second, we remark that the correspondence
between the Felderhof models and Schur functions \cite{BBF}
was generalized to factorial Schur functions \cite{BMN}.
We furthermore generalized this correspondence
and find the following:
First we introduce the following generalized $L$-operator
for the Felderhof model
\begin{eqnarray}
L_{aj}(z,t,\alpha_j,\gamma_j)=\left( 
\begin{array}{cccc}
1-\gamma_j z & 0 & 0 & 0 \\
0 & t+\gamma_j z & 1 & 0 \\
0 & (t+1)z & \alpha_j+(1-\alpha_j \gamma_j)z & 0 \\
0 & 0 & 0 & -t \alpha_j+(1-\alpha_j \gamma_j)z
\end{array}
\right), \label{generalizedfelderhofloperator}
\end{eqnarray}
and introduce the $N$-particle state
\begin{align}
\Phi(\{ z \}_N,t,\{ \alpha \},\{ \gamma \}) \rangle
=
B(z_1,t,\{ \alpha \},\{ \gamma \}) \cdots
B(z_N,t,\{ \alpha \},\{ \gamma \})| \Omega \rangle,
\end{align}
where
\begin{align}
B(z,t,\{ \alpha \},\{ \gamma \})
={}_a \langle 0 |
L_{a N}(z,t,\alpha_N,\gamma_N) \cdots
L_{a 1}(z,t,\alpha_1,\gamma_1)| 1 \rangle_a.
\end{align}
We found the correspondence between the wavefunction
$\langle x_1 \dots x_N|\Phi(\{ z \}_N,t,\{ \alpha \},\{ \gamma \}) \rangle$
and the generalized factorial Schur functions defined below.
\begin{definition}
We define the generalized factorial Schur functions
to be the following determinant:
\begin{align}
s_\lambda({\bf z}|\{ \alpha \}|\{ \gamma \})
=\frac{F_{\lambda+\delta}({\bf z}|\{ \alpha \}|\{ \gamma \})}
{\prod_{1 \le j < k \le N}(z_j-z_k)},
\end{align}
where ${\bf z}=\{z_1,\dots,z_N \}$ is a set of variables
and $\lambda$ denotes a Young diagram
$\lambda=(\lambda_1,\lambda_2,\dots,\lambda_N)$
with weakly decreasing non-negative integers
$\lambda_1 \ge \lambda_2 \ge \cdots \ge \lambda_N \ge 0$,
and $\delta=(N-1,N-2,\dots,0)$.
$F_{\mu}({\bf z}|\{ \alpha \}|\{ \gamma \})$
is an $N \times N$ determinant
\begin{align}
F_{\mu}({\bf z}|\{ \alpha \}| \{ \gamma \})
=\mathrm{det}_N
(
f_{\mu_j}(z_k|\{ \alpha \}|\{ \gamma \})
),
\end{align}
where
\begin{align}
f_\mu(z|\{ \alpha \}|\{ \gamma \})
=\prod_{j=1}^\mu(\alpha_j+(1-\alpha_j \gamma_j)z)
\prod_{j=\mu+2}^M(1-\gamma_j z).
\end{align}
\end{definition}

\begin{theorem}
The wavefunction
$\langle x_1 \dots x_N|\Phi(\{ z \}_N,t,\{ \alpha \},\{ \gamma \}) \rangle$
is expressed by the generalized factorial Schur functions
$s_\lambda({\bf z}|\{ \alpha \}|\{ \gamma \})$ as
\begin{align}
\langle x_1 \dots x_N|\Phi(\{ z \}_N,\{ \alpha \},\{ \gamma \}) \rangle
=\prod_{1 \le j<k \le N}(z_j+tz_k)
s_\lambda({\bf z}|\{ \alpha \}|\{ \gamma \}),
\end{align}
under the relation $\lambda_j=x_{N-j+1}-N+j-1$, $j=1,\dots,N$.
\end{theorem}

Based on this correspondence, one can show the following
dual Cauchy formula for the generalized factorial Schur functions.

\begin{theorem}
The following dual Cauchy formula holds
for the generalized factorial Schur functions
with sets variables ${\bf x}=\{x_1,\dots,x_N \}$, ${\bf y}=\{y_1,\dots,y_M \}$,
$\{ \alpha \}=\{\alpha_1,\dots,\alpha_{N+M} \}$,
$\{ \gamma \}=\{\gamma_1,\dots,\gamma_{N+M} \}$,
\begin{align}
\sum_{\lambda \subseteq M^N}s_\lambda({\bf x}|\{\alpha \}|\{\gamma \})
s_{\hat{\lambda}}({\bf y}|\{-\alpha \}|\{-\gamma \})
=\prod_{j=1}^N \prod_{k=1}^M (x_j+y_k)
\prod_{1 \le j<k \le N+M}(1+\alpha_j(\gamma_k-\gamma_j)),
\end{align}
where
$\{ -\alpha \}=\{-\alpha_1,\dots,-\alpha_{N+M} \}$,
$\{ -\gamma \}=\{-\gamma_1,\dots,-\gamma_{N+M} \}$, and
$\hat{\lambda}
=(\hat{\lambda}_1,\dots,\hat{\lambda}_M)$ is the partition
of the Young diagram $\lambda=(\lambda_1,\dots,\lambda_N)$ given by
\begin{align}
\hat{\lambda}_i=|\{j \ | \ \lambda_j  \le M-i \}|.
\end{align}
\end{theorem}

It is also known that changing the boundary condition
of the partition functions
to the Tsuchiya boundary condition \cite{Ts},
which is a class of partition function with reflecting end,
what appears is the symplectic Schur functions
\cite{Iv,BBCG}.
We have also generalized this to a more general setting,
and find the following dual Cauchy formula for the
generalized factorial symplectic Schur functions.

\begin{definition}
We define the generalized symplectic Schur functions to be the
following determinant:
\begin{align}
sp_\lambda({\bf z}|\{ \widetilde{\alpha} \}|\{ \widetilde{\gamma} \})
=\frac{G_{\lambda+\delta}({\bf z}|\{ \widetilde{\alpha} \}|\{ 
\widetilde{\gamma} \})}
{\mathrm{det}_N(z_k^{N-j+1}-z_k^{-N+j-1})}.
\end{align}
Here, $G_{\mu}({\bf z}|\{ \widetilde{\alpha} \}|\{ \widetilde{\gamma} \})$
is an $N \times N$ determinant
\begin{align}
G_{\mu}({\bf z}|\{ \widetilde{\alpha} \}|\{ \widetilde{\gamma} \})
=\mathrm{det}_N
(
g_{\mu_j}(z_k|\{ \widetilde{\alpha} \}|\{ \widetilde{\gamma} \})
-
g_{\mu_j}(z_k^{-1}|\{ \widetilde{\alpha} \}|\{ \widetilde{\gamma} \})
),
\end{align}
where
\begin{align}
g_\mu(z|\{ \widetilde{\alpha} \}|\{ \widetilde{\gamma} \})
=\prod_{j=0}^\mu(\alpha_j+(1-\alpha_j \gamma_j)z)
\prod_{j=\mu+2}^M(1-\gamma_j z)
\prod_{j=1}^M(1-\gamma_jz^{-1}).
\end{align}
\end{definition}

\begin{theorem}
The following dual Cauchy formula holds
for the generalized factorial symplectic Schur functions
with sets variables ${\bf x}=\{x_1,\dots,x_N \}$, ${\bf y}=\{y_1,\dots,y_M \}$,
$\{ \widetilde{\alpha} \}=\{\alpha_0,\dots,\alpha_{N+M} \}$,
$\{ \widetilde{\gamma} \}=\{\gamma_0,\dots,\gamma_{N+M} \}$,
\begin{align}
&\sum_{\lambda \subseteq M^N}sp_\lambda({\bf x}|\{ \widetilde{\alpha} \}|\{
\widetilde{\gamma} \})
sp_{\hat{\lambda}}({\bf y}|\{- \widetilde{\alpha} \}|\{- \widetilde{\gamma} \}) \nonumber \\
=&\prod_{j=1}^M y_j^{-N}
\prod_{j=1}^N \prod_{k=1}^M (1+x_j y_k)(1+x_j^{-1} y_k)
\prod_{0 \le j<k \le N+M}(1+\alpha_j(\gamma_k-\gamma_j))
\prod_{1 \le j<k \le N+M}(1-\gamma_j \gamma_k),
\end{align}
where
$\{ -\widetilde{\alpha} \}=\{-\alpha_0,\dots,-\alpha_{N+M} \}$,
$\{ -\widetilde{\gamma} \}=\{-\gamma_0,\dots,-\gamma_{N+M} \}$
and
$\hat{\lambda}
=(\hat{\lambda}_1,\dots,\hat{\lambda}_M)$ is the partition
of the Young diagram $\lambda=(\lambda_1,\dots,\lambda_N)$ given by
\begin{align}
\hat{\lambda}_i=|\{j \ | \ \lambda_j  \le M-i \}|.
\end{align}
\end{theorem}
Details of the theorems on this subsection will be given elsewhere \cite{MSW}.

\section{Conclusion}
In this paper, we reviewed and presented new results on
the relation between integrable lattice models and
combinatorial representation theory of symmetric polynomials.

The philosophy of statistical mechanics
is to investigate macroscopic quantities from
microscopic description of the system.
For the case of integrable lattice models,
this means to study symmetric polynomials
as partition functions constructed
from the local $L$-operators and global boundary conditions.
The correspondence between symmetric polynomials
and partition functions
 allows us to investigate
and find various new combinatorial formulae which seems
to be extremely hard to discover or to prove
without integrable models.

Some of the correspondence is essentially equivalent to the notions
in the field of algebraic combinatorics,
especially Schubert calculus.
To make further progresses on the representation theory
of symmetric polynomials, we think the point of view
of quantum integrability is indispensable.
For example, the wavefunction \eqref{qbetaGrothendieckpolynomials}
is ``the" (not just ``a")
$q$-deformation of the $\beta$-Grothendieck polynomials,
which the parameter $q$ is nothing but the parameter
for the quantum group $U_q(sl_2)$,
and we believe quantum integrability is essential
to find this deformation.
We expect further advances will be made from
the interplay between quantum integrable models
and representation theory in the future.

\section*{Acknowledgments}
This work was partially supported by
grant-in-Aid for Research Activity start-up
No. 15H06218 and
Scientific Research (C) No. 24540393.

\end{document}